\documentclass{article}

\usepackage[T1]{fontenc}
\usepackage[utf8]{inputenc}
\usepackage[includefoot,margin=3cm,a4paper]{geometry}
\usepackage{microtype}

\usepackage{algorithm}
\usepackage{algorithmic}

\usepackage{natbib}
\usepackage{graphicx} %
\usepackage{mathrsfs}
\usepackage{paralist}
\usepackage[super]{nth}
\usepackage[toc,page]{appendix}
\usepackage{listings}
\usepackage{mathtools}
\usepackage{color}
\usepackage{tabularx}
\usepackage{amsmath}
\usepackage{amsthm}
\usepackage{csquotes}
\usepackage{multicol}
\usepackage{amssymb}
\usepackage{fancyvrb}
\usepackage{makecell}
\usepackage{dsfont}
\usepackage{graphicx} %
\usepackage{mathrsfs}
\usepackage{amsthm}
\usepackage[super]{nth}
\usepackage[toc,page]{appendix}
\usepackage{listings}
\usepackage{mathtools}
\usepackage{color}
\usepackage{tabularx}
\usepackage{amsmath}
\usepackage{csquotes}
\usepackage{multicol}
\usepackage{amssymb}
\usepackage{fancyvrb}
\usepackage{makecell}
\usepackage{dsfont}
\usepackage{url}

\usepackage{mdframed}
\usepackage[toc,page]{appendix}
\usepackage{mathdots}
\usepackage{xspace}
\usepackage{empheq}
\usepackage{arydshln}
\usepackage{nicefrac}
\usepackage{dsfont}
\usepackage[pdfdisplaydoctitle,menucolor=orange!40!black,filecolor=magenta!40!black,urlcolor=blue!40!black,linkcolor=red!40!black,citecolor=green!40!black,ocgcolorlinks, breaklinks = true]{hyperref} %
\usepackage{breakcites}

\usepackage[capitalise]{cleveref}
\crefname{figure}{Figure}{Figures}

\usepackage{array}
\newcolumntype{L}[1]{>{\raggedright\let\newline\\\arraybackslash\hspace{0pt}}m{#1}}
\newcolumntype{C}[1]{>{\centering\let\newline\\\arraybackslash\hspace{0pt}}m{#1}}
\newcolumntype{R}[1]{>{\raggedleft\let\newline\\\arraybackslash\hspace{0pt}}m{#1}}

\usepackage{graphicx} %
\usepackage{mathrsfs}
\usepackage{enumitem}
\usepackage{amsthm}
\usepackage[super]{nth}
\usepackage[toc,page]{appendix}
\usepackage{listings}
\usepackage{mathtools}
\usepackage{color}
\usepackage{tabularx}
\usepackage{amsmath}
\usepackage{csquotes}
\usepackage{multicol}
\usepackage{amssymb}
\usepackage{fancyvrb}
\usepackage{makecell}

\usepackage{mdframed}
\usepackage[toc,page]{appendix}
\usepackage{mathdots}
\usepackage{xspace}
\usepackage{empheq}
\usepackage{arydshln}
\usepackage{nicefrac}
\usepackage{thm-restate}
\usepackage{cleveref}
\newcommand{\clrstr}{20} 
\allowdisplaybreaks

\usepackage{array}
\newcolumntype{L}[1]{>{\raggedright\let\newline\\\arraybackslash\hspace{0pt}}m{#1}}
\newcolumntype{C}[1]{>{\centering\let\newline\\\arraybackslash\hspace{0pt}}m{#1}}
\newcolumntype{R}[1]{>{\raggedleft\let\newline\\\arraybackslash\hspace{0pt}}m{#1}}

\usepackage{tikz}
\usetikzlibrary{automata, positioning, calc, shapes, arrows, fit}

\DeclareMathOperator{\sw}{\textsf{uw}}
\DeclareMathOperator{\cov}{\textsf{cov}}

\DeclareMathOperator{\uA}{\overline{s}}
\newcommand{\AVmax}{W^*}
\DeclareMathOperator{\supp}{mms}
\newcommand{\AV}[1]{\overline{AV}(#1)}
\newcommand{\CC}[1]{\overline{CC}(#1)}
\newcommand{\pricea}{affordable\xspace}
\theoremstyle{plain}
\newtheorem{theorem}{Theorem}

\newtheorem{example}{Example}

\newtheorem{definition}{Definition}

\title{Completing Priceable Committees:\\
Utilitarian and Representation Guarantees for Proportional Multiwinner Voting
}

\author{}
\date{\today}

\usepackage[textsize=tiny]{todonotes}
\title{Completing Priceable Committees:\\
Utilitarian and Representation Guarantees for Proportional Multiwinner Voting}
\author{Markus Brill\footnote{University of Warwick, UK; email: markus.brill@warwick.ac.uk} \and Jannik Peters\footnote{TU Berlin, Germany; email: jannik.peters@tu-berlin.de} }
\date{}

\begin{document}

\maketitle

\begin{abstract}
When selecting committees based on preferences of voters, a variety of different criteria can be considered. Two natural objectives are maximizing the \textit{utilitarian welfare} (the sum of voters' utilities) and \textit{coverage} (the number of represented voters) of the selected committee. Previous work has studied the impact on utilitarian welfare and coverage when requiring the committee to satisfy minimal requirements such as \textit{justified representation} or \textit{weak proportionality}. In this paper, we consider the impact of imposing much more demanding proportionality axioms. We identify a class of voting rules that achieve strong guarantees on utilitarian welfare and coverage when combined with appropriate completions. This class is defined via a weakening of priceability and contains prominent rules such as the Method of Equal Shares. We show that committees selected by these rules (i) can be completed to achieve optimal coverage and (ii) can be completed to achieve an asymptotically optimal approximation to the utilitarian welfare if they additionally satisfy EJR+. Answering an open question of Elkind et al. (2022), we use the \textit{Greedy Justified Candidate Rule} to obtain the best possible utilitarian guarantee subject to proportionality. We also consider completion methods suggested in the participatory budgeting literature and other objectives besides welfare and coverage.
\end{abstract}

\section{Introduction}

In multiwinner voting, a subset of candidates (often called a \textit{committee}) needs to be chosen in a way that reflects the preferences of a set of voters over these candidates. The (computational) social choice literature has identified many different, often competing, criteria for committees \citep{FeMa92b,EFSS17a,LaSk22a}. In particular, \citet{FSST17a} distinguish between three important desiderata referred to as \textit{individual excellence}, \textit{diversity}, and \textit{proportional representation}. 

This paper focuses on \textit{approval-based} multiwinner voting, where voters cast approval ballots over individual candidates \citep{LaSk22a}. In this setting, it is commonly assumed that each voter evaluates a committee by counting the number of approved committee members. This number is often referred to as the \textit{utility} that the voter derives from the committee, and a voter is said to be \textit{represented} by a committee if their utility is nonzero. 
Using these conventions, the first two desiderata of \citet{FSST17a} have straightforward interpretations:
Individual excellence is measured by the \textit{utilitarian welfare} of a committee, i.e., the sum over the voters' utilities. Utilitarian welfare is maximized when choosing the candidates with the highest number of approvers, and the voting rule that selects this committee is known as \textit{Approval Voting (AV)}.
The diversity of a committee is captured by its \textit{coverage}, i.e., the number of voters that are represented by it. 
The (computationally intractable) voting rule that selects committees maximizing coverage is known as \textit{Chamberlin--Courant (CC)}. 
The third desideratum, proportional representation, is much harder to formalize, and the literature has identified a host of different (axiomatic and quantitative) measures to assess the proportionality of a committee \citep{ABC+16a, SFF+17a, Skow21a, PPS21a, BrPe23a}; a prominent example is the \textit{extended justified representation (EJR)} axiom.  

The study of trade-offs between the three desiderata has been initiated by \citet{LaSk20b}, who analyzed how closely classic multiwinner rules such as PAV and Phragmén's rules 
approximate the optimal utilitarian welfare and coverage. The resulting bounds are referred to as \textit{utilitarian} and \textit{representation guarantees}, respectively. Besides guarantees for specific rules, they also proved upper bounds on the utilitarian and representation guarantees of rules satisfying minimal proportionality requirements. In their follow-up work, \citet{EFI+22a} studied which guarantees can be achieved if the committee is required to satisfy the \textit{justified representation (JR)} axiom. While \citet{LaSk20b} gave an upper bound of 
${2}/{\lfloor \sqrt{k}\rfloor} - {1}/{k}$
for the utilitarian guarantee of any committee satisfying
JR\,---\,with a higher guarantee being better\,---\,\citet{EFI+22a} were almost able to close this gap by showing that a utilitarian guarantee of $\frac{2-\varepsilon}{1 + \sqrt{k}}$ is possible for any $\varepsilon > 0$ and sufficiently large $k$.\footnote{
Here, $k$ denotes the size of the committee.
\citet{LaSk20b} proved the upper bound for an axiom they call \textit{weak proportionality}, but\,---\,as \citet{EFI+22a} observe\,---\,the bound also holds for JR. Moreover, we note that \citet{EFI+22a} phrase their bounds in terms of the ``price of JR,'' with lower prices being better: a price of $P$ corresponds to a guarantee of $1/P$.}
Further, they showed that a representation guarantee of $\frac{3}{4}$ is possible in conjunction with EJR and that both guarantees can be simultaneously approximated close to optimal while requiring JR. They achieve the bounds for JR by using a variant of the \textit{sequential Chamberlin--Courant} rule,
and the bound for EJR is achieved by the somewhat unnatural \textit{GreedyEJR} rule. Both of these rules construct committees by matching committee members to voters approving them (see \Cref{sec:affordable} for details), in a way that is reminiscent of \textit{priceability} and rules such as the \textit{Method of Equal Shares (MES)} \citep{PeSk20b}. These rules suffer from a common drawback: They often select strictly fewer than the desired number of candidates, and thus require a \textit{completion method} to output a full committee.

\paragraph{Our Contribution}

In this paper, we study which representation and utilitarian guarantees are possible when requiring much more demanding proportionality axioms than those considered by \citet{LaSk20b} and \citet{EFI+22a}.
We identify a class of voting rules that achieve strong guarantees on utilitarian welfare and coverage when combined with appropriate completion methods. 
This class, defined via a weakening of priceability \citep{PeSk20b} we call \textit{affordability}, contains prominent voting rules such as MES. %

We show that any affordable committee, when completed with the CC rule, achieves an optimal representation guarantee of $\frac{3}{4}$. This general result has several advantages over earlier results on representation guarantees: First, our result implies that any proportionality axiom which is compatible with affordability, such as EJR+ \citep{BrPe23a} or FJR \citep{PPS21a}, is also compatible with a representation guarantee of $\frac{3}{4}$. Second, our result uses much simpler rules: Instead of the unnatural and computationally intractable GreedyEJR rule, we can use attractive rules such as MES to achieve this bound. 

For utilitarian guarantees, we show that any affordable committee that additionally satisfies EJR+ (which is more demanding than EJR but still satisfied by MES), completed with AV, achieves a utilitarian guarantee of $\Omega(\frac{1}{\sqrt{k}})$. We further show that the recently introduced \textit{Greedy Justified Candidate Rule (GJCR)} \citep{BrPe23a}, completed with AV, achieves a utilitarian guarantee of exactly $\frac{2}{\sqrt{k}} - \frac{1}{k}$, which is the best possible guarantee for proportional rules. This answers an open question by \citet{EFI+22a}, who were only able to obtain this bound asymptotically. 
Moreover, we establish an interesting distinction between EJR and EJR+, by showing that GreedyEJR (which satisfies EJR but not EJR+) does not achieve a utilitarian guarantee of $\Omega(\frac{1}{k^c})$ for any $c < 1$, even on instances where it is exhaustive, i.e., where it selects exactly $k$ candidates. %
Thus, EJR is not sufficient for strong utilitarian guarantees.  

We also study trade-offs between utilitarian and representation guarantees and show that one can simultaneously achieve a representation guarantee of $\frac{3}{4}-o(1)$ and a utilitarian guarantee of $\frac{2}{\sqrt{k}} - \frac{1}{k}$ together with EJR+, 
improving and simplifying a similar result by \citet{EFI+22a} for~JR.

Finally, we consider completion methods that are relevant for specific applications of multiwinner voting: completion by the maximin support method \citep{SFFB18a}, which is relevant for the security of blockchain systems \citep{CeSt21a}, and completion by perturbation or by varying the budget, which have been discussed in the participatory budgeting (PB) literature \citep{ReMa23a}.

\paragraph{Related Work}

Besides the papers by \citet{LaSk22a} and \citet{EFI+22a}, utilitarian and representation guarantees have been studied 
in the context of PB by \citet{FVMG22a}, who proved some guarantees, but mostly impossibility results.
The topic of completing non-exhaustive voting rules is especially relevant in PB, where various ways to complete the Method of Equal Shares have been discussed and analyzed on real-world data \citep{PPS21a, FFP+23a, BFJK23a}. One of the completion methods we study in this paper, ``completion by varying the budget,'' was recently used in a real-world PB election in the Polish city of Wielizka (see \url{https://equalshares.net} for details).

In a similar spirit as the paper by \citet{EFI+22a}, \citet{MPS20a} and \citet{TWZ20a} studied the ``price of fairness'' in the setting of budget division (or probabilistic social choice).

\section{Preliminaries}

We are given a finite set $C = \{c_1, \dots, c_m\}$ of $m>0$ \emph{candidates} and a finite set $N = \{1, \dots, n\}$ of $n>0$ \emph{voters}. Voters cast \emph{approval ballots} over candidates: for each voter $i \in N$, the set $A_i \subseteq C$ consists of the candidates that are {approved} by voter $i$. Together $A = (A_i)_{i \in N}$ forms an \emph{approval profile}. For a candidate $c \in C$, we call the members of the set $N_c = \{i \in N \colon c \in A_i\}$ the \emph{approvers} of $c$.  Further, we are given a \textit{committee size} $k \le m$. We call any set $W \subseteq C$ of size $\lvert W \rvert \le k$ a \emph{committee}. If $\lvert W \rvert = k$, we say that $W$ is \emph{exhaustive}. Finally, we call $\mathcal{I} = (A, C, k)$ an \emph{instance}. 

A \emph{(multiwinner voting) rule} $r$ is a function mapping instances $\mathcal{I} = (A, C, k)$ to non-empty sets of committees $r(\mathcal{I})$. We often say that a rule ``satisfies'' a property if, for each instance $\mathcal{I}$, all committees in $r(\mathcal{I})$ satisfy the property.

\paragraph{Proportionality Axioms}
First, we consider the classical justified representation axioms of \citet{ABC+16a}. For an integer $\ell \in \mathbb{N}$, a group $N' \subseteq N$ of voters is said to be \emph{$\ell$-large} if $\lvert N' \rvert \ge \ell  \frac{n}{k}$ and \emph{$\ell$-cohesive} if $\left\lvert \bigcap_{i \in N'} A_i \right\rvert \ge \ell$. The \emph{justified representation (JR)} axiom requires that for any $1$-large and $1$-cohesive group $N'$, there exists some $i \in N'$ with $\lvert A_i \cap W \rvert \ge 1$. \emph{Extended Justified Representation (EJR)} requires that for all $\ell \in [k]$ and for any $\ell$-large and $\ell$-cohesive group $N'$ there exists some $i \in N'$ with $\lvert A_i \cap W \rvert \ge \ell$. Finally, the recently introduced \emph{EJR+} axiom requires that for any $\ell$-large and $1$-cohesive group $N'$, there either exists some $i \in N'$ with $\lvert A_i \cap W \rvert \ge \ell$ or it holds that $\bigcap_{i \in N'} A_i \subseteq W$, for all $\ell \in [k]$. It is easy to see that EJR+ implies EJR which in turn implies JR. 
We also use the following quantitative notion \citep{BrPe23a}.

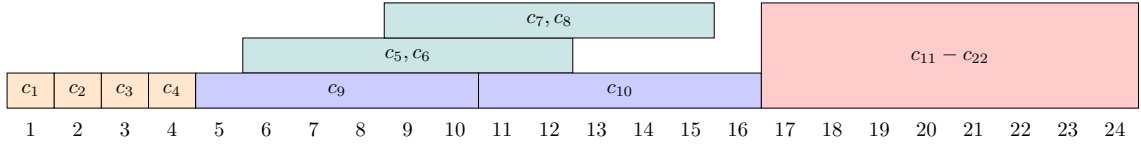
\begin{figure*}
    \centering
        \scalebox{0.775}{
    \begin{tikzpicture}
    [yscale=0.6,xscale=0.8,
    voter/.style={anchor=south}]
    
        \foreach \i in {1,...,24}
    		\node[voter] at (\i-0.5, -1) {$\i$};
        
        \draw[fill=orange!\clrstr] (0, 0) rectangle (1, 1);
        \draw[fill=orange!\clrstr] (1, 0) rectangle (2, 1);
        \draw[fill=orange!\clrstr] (2, 0) rectangle (3, 1);
        \draw[fill=orange!\clrstr] (3, 0) rectangle (4, 1);
        \draw[fill=teal!\clrstr] (5, 1) rectangle (12, 2);
        \draw[fill=teal!\clrstr] (8, 2) rectangle (15, 3);
        \draw[fill=blue!\clrstr] (4, 0) rectangle (10, 1);
        \draw[fill=blue!\clrstr] (10, 0) rectangle (16, 1);
        \draw[fill=red!\clrstr] (16, 0) rectangle (24, 3);

        \node at ( 0.5, 0.5) {$c_{1}$};
        \node at ( 1.5, 0.5) {$c_{2}$};
        \node at ( 2.5, 0.5) {$c_{3}$};
        \node at ( 3.5, 0.5) {$c_{4}$};
        \node at ( 8.5, 1.5) {$c_{5}, c_{6}$};
        \node at ( 11.5, 2.5) {$c_{7}, c_{8}$};
        \node at ( 7, 0.5) {$c_{9}$};
        \node at ( 13, 0.5) {$c_{10}$};
        \node at ( 20, 1.5) {$c_{11}-c_{22}$};
    \end{tikzpicture}}
    \caption{Instance with $n=24$ voters and $m=22$ candidates considered in \Cref{exp:rules}. Voters correspond to integers and approve candidates placed above them, e.g.,  voter $1$ approves candidate $c_1$ and voter $9$ approves candidates $c_5, c_6, c_7, c_8, c_9$. }
    \label{fig:example}
\end{figure*}

\begin{definition}
    Given a committee $W$,  a candidate $c \in C\setminus W$ is \emph{$(t, \ell )$-represented} if there is no $\ell$-large group $N' \subseteq N_c$ with $\frac{1}{\lvert N'\rvert} \sum_{i \in N'} \lvert A_i \cap W \rvert < t$. For a given function \mbox{$f \colon \! \mathbb{N} \!\to \! \mathbb{R}$}, a committee is \emph{$f$-representative} if every unchosen candidate $c \in C \setminus W$ is $(f(\ell), \ell)$-represented for any $\ell \in [k]$.
\end{definition}
This is a stronger notion than the \textit{proportionality degree} \citep{Skow21a}, in the sense that an $f$-representative committee (or rule) also has a proportionality degree of $f$, as the definition just removes the cohesiveness requirement inherent to the proportionality degree. The representativeness of a committee is polynomial-time computable, and as we show later, contrary to the proportionality degree, 
it allows us to obtain bounds on the utilitarian guarantee of a voting rule. %

\paragraph{Rules}
Given an instance $\mathcal{I}=(A,C,k)$ and a committee $W \subseteq C$, we let 
$\sw(W) = \sum_{i \in N} \lvert A_i \cap W\rvert$
denote the \emph{utilitarian welfare} of $W$ and $\cov(W) = \lvert \{i \in N \colon A_i \cap W \neq \emptyset\} \rvert$ the \emph{coverage} of $W$.  
The \emph{Approval Voting (AV)} rule selects a committee of size $k$ maximizing the utilitarian welfare. %
The \emph{Chamberlin--Courant (CC)} rule selects a committee of size $k$ maximizing the coverage. %
The \emph{sequential Chamberlin--Courant (seq-CC)} rule  starts with an empty committee $W = \emptyset$ and greedily adds candidates $c \notin W$ maximizing $\cov(W \cup \{c\})$ until $W$ has size $k$.

Next, we define three rules we study throughout the paper. 
The \emph{Method of Equal Shares} (MES) \citep{PeSk20b} starts with an empty committee $W = \emptyset$ and with a budget $b_i = \frac{k}{n}$ for each voter $i \in N$.
Then, for each unselected candidate $c \notin W$, it calculates the value $\rho$ such that 
$\sum_{i \in N_c} \min(b_i, \rho) = 1$. 
If such a $\rho$ does not exist for any unselected candidate, MES terminates.  
Otherwise, it selects the candidate $c$ with the minimum $\rho$, adds it to $W$, and sets $b_i = b_i - \min(b_i, \rho)$ for each $i \in N_c$. 

The \emph{Greedy Justified Candidate Rule} (GJCR) \citep{BrPe23a} 
also starts with an empty committee $W = \emptyset$. GJCR then repeatedly checks if there is a candidate $c$, set $ N' \subseteq N_c $, and $\ell \in [k]$ with $\lvert N' \rvert \ge \frac{\ell n}{k}$ and $\lvert A_i \cap W\rvert < \ell$ for all $i \in N'$. If there is such a candidate, it picks the candidate maximizing $\lvert N'\rvert$ and adds it to $W$, otherwise it terminates.

The \emph{GreedyEJR} rule %
\citep{BFNK19a, PPS21a}  
works similarly to GJCR, but iterates over voter groups rather than candidates (which makes the rule computationally intractable). 
Starting with an empty committee $W$, GreedyEJR finds the largest $\ell$ such that there is an $\ell$-cohesive and $\ell$-large group $N'$, adds~$\ell$ candidates from $\bigcap_{i \in N'} A_i$ to $W$, and deletes all voters in $N'$; this is repeated until no further group can be found.  

MES and GJCR satisfy EJR+ \citep{BrPe23a}, whereas GreedyEJR satisfies EJR but not EJR+. All three rules may return non-exhaustive committees (\Cref{exp:rules}).

\paragraph{Utilitarian and Representation Guarantees}
For a given instance $\mathcal{I}$, we define the \textit{utilitarian ratio} of a committee $W$ as the fraction of the highest achievable utilitarian welfare in this instance, i.e.,
$\frac{\sw(W)}{\sw(AV(\mathcal{I}))}$. Analogously, $\frac{\cov(W)}{\cov(CC(\mathcal{I}))}$ is called the \textit{representation ratio} of $W$.

A utilitarian or representation \emph{guarantee} is a lower bound on the utilitarian or representation ratio of a committee or a set of committees (such as those output by a rule).

\paragraph{Completion Methods} 
For a non-exhaustive committee~$W$, we let $\AV{W}$ denote the completion of $W$ with the highest approval score. Formally, $\AV{W} = W \cup AV(\mathcal{I'})$, where $\mathcal{I'}=(A,C\setminus W,k-|W|)$.
Further, $\CC{W}$ is the completion of $W$ with the highest coverage. Formally,
$\CC{W} = W \cup CC(\mathcal{I''})$, where $\mathcal{I''}=(A',C\setminus W,k-|W|)$ and $A'$ 
 results from $A$ by deleting all voters represented by~$W$.
Completion via seq-CC is defined analogously.

\begin{example}
    Consider the instance depicted in \Cref{fig:example} with $n = 24, k = 12$, and $22$ candidates. 
    In this instance, AV selects the candidates $\{c_{11}, \dots, c_{22}\}$ to achieve a utilitarian welfare of $12 \cdot 8 = 96$. This committee, however, does not satisfy EJR (or even JR), as for instance voters $5$ to $10$ together with candidate $c_9$ witness a violation of JR (these voters deserve $6 \frac{k}{n} = 3$ seats on the committee). The optimal coverage is $n = 24$ and can be achieved, e.g., by $\{c_1, \dots, c_{12}\}$. 

    \textbf{MES} begins by assigning everyone a budget of $\frac{1}{2}$. Four candidates out  of $c_{11}$ to $c_{22}$ would be bought first, with a $\rho$ of $\frac{1}{8}$ each. Then, $c_5, c_6, c_7$ can be bought with $\rho = \frac{1}{7}$ each. The approvers of $c_8$ would then have a budget of $4 (\frac{1}{2} - \frac{3}{7}) + 3(\frac{1}{2} - \frac{2}{7}) = \frac{13}{14} < 1$, which is not enough to buy $c_8$. Both the approvers of $c_9$ and of $c_{10}$, on the other hand, have enough budget left to buy their candidate. The approvers of any of $c_1$ to $c_4$ have a budget of $\frac{1}{2}$ each and therefore, can not buy the candidate they approve. Hence, MES would (up to tie-breaking) return the committee $W = \{c_5, c_6, c_7, c_9, \dots, c_{14}\}$ with a utilitarian welfare of %
    $2 \cdot 1 + 3 \cdot 2 + 3 \cdot 3 + 12 \cdot 4 = 65$
    and a utilitarian ratio of
    $\frac{65}{96}$. Since it covers all but $4$ voters, its representation ratio is $\frac{20}{24} = \frac{5}{6}$. The completion $\CC{W}$ additionally contains three of $c_1, c_2, c_3$, leading to a representation ratio of $\frac{23}{24}$ and a utilitarian ratio of $\frac{68}{96}$, while the completion $\AV{W}$ would contain three of $c_{14}$ to $c_{22}$, leading to a utilitarian ratio of $\frac{86}{96}$ and an (unchanged) representation ratio of $\frac{5}{6}$.

    \smallskip
    \textbf{GJCR} could select $\{c_5, c_6, c_7, c_{10}, c_{11}, c_{12}, c_{13}, c_{14}\}$ (notably leaving out $c_9$, but still selecting $c_{10}$ since the voters $13$ to $16$ deserve two seats in the committee). Thus, GJCR completed with CC selects $c_1, \dots, c_4$, leading to a perfect representation ratio of $1$, while GJCR completed with AV selects four out of $c_{15}$ to $c_{22}$, reaching a utilitarian ratio of~$\frac{91}{96}$.

    \smallskip
    \textbf{GreedyEJR} could select $\{c_5, c_6, c_{10}, c_{11}, c_{12}, c_{13}, c_{14}\}$ by first considering the $4$-cohesive group $17$ to $24$, then the $2$-cohesive group $6$ to $12$ with candidates $c_5$ and $c_6$, and finally the $1$-cohesive group $13$ to $16$ with candidate $c_{10}$. Notably, this committee violates EJR+ since the approvers of $c_9$ would deserve $3$ seats, but get at most $2$. Since they do not form a cohesive group, they could not get selected by GreedyEJR and therefore GreedyEJR does not satisfy EJR+.
    \label{exp:rules}
\end{example}

\section{Affordable Committees}
\label{sec:affordable}

The three rules that play an important role in this paper are the Method of Equal Shares, the Greedy Justified Candidate Rule, and GreedyEJR. A common feature of these three rules is that they\,---\,explicitly or implicitly\,---\,construct (fractional) matchings between committee members and their approvers: each voter starts off with a ``voting weight'' of $k/n$ 
(so that the total voting weight per committee member is~$1$), and each time a candidate is added to the committee under one of these rules, there is $1$ unit of (previously unused) voting weight supporting this candidate, which then gets used up by the candidate.

This is reminiscent of the notion of \textit{priceability} \citep{PeSk20b}, which assumes that each voter owns an equal amount of ``money'' that can be spent on candidates they approve. %
A committee is called {priceable} if it is possible to allocate the money 
in such a way that each committee member is paid for and there is not enough leftover money to pay for the inclusion of further candidates.\footnote{Priceability has been studied in various forms \citep{PPSS21a,MSW22a, BFL+23a, LaMa23a,BrPe23a, KrEl23b}.}   

To formalize this, consider an instance $\mathcal{I} = (A, C, k)$. A \textit{payment system} is a collection $(p_i)_{i \in N}$ containing a function $p_i \colon C \to \mathbb{R}_{\ge 0}$ for each voter $i \in N$. A committee \mbox{$W\subseteq C$} is said to be \emph{priceable} if there exists a payment system $(p_i)_{i \in N}$ satisfying the following constraints:
\begin{itemize}
    \item[\textbf{C1}] $p_i(c) = 0$ if $c \notin A_i$ for all $c \in C$ and $i \in N$
    \item[\textbf{C2}] $\sum_{c \in C} p_i(c) \le \frac{k}{n}$ for all $i \in N$
    \item[\textbf{C3}] $\sum_{i \in N} p_i(c) = 1$ for all $c \in W$
    \item[\textbf{C4}] $\sum_{i \in N} p_i(c) = 0$ for all $c \notin W$
    \item[\textbf{C5}] $\sum_{i \in N_c} \left(\frac{k}{n} - \sum_{c \in C} p_i(c)\right) \le 1$ for all $c \notin W$.    
\end{itemize}

\noindent Note that the left-hand side of constraint \textbf{C5} exactly corresponds to the money leftover with approvers of candidate~$c$.

MES is probably the most prominent example of a priceable rule (i.e., a rule that always outputs priceable committees).
GJCR and GreedyEJR, on the other hand, do not necessarily output priceable committees as they might violate \textbf{C5} (while still satisfying \textbf{C1}--\textbf{C4}). As it turns out, constraints \textbf{C1}--\textbf{C4} are sufficient to prove strong guarantees. We call such committees affordable. 

\begin{definition}
A committee $W\subseteq C$ is \emph{affordable} if there is a payment system for $W$ satisfying constraints \textbf{C1}--\textbf{C4}.
\end{definition}

Clearly, every priceable committee is affordable. Affordability is a rather weak requirement: it simply requires that each committee member can be associated with a sufficient amount of voter support. Every subcommittee of an affordable committee is affordable, and so is the empty committee. 

Affordability (and thus priceability) is not generally compatible with exhaustiveness. For example, consider an instance with two voters and two candidates such that each candidate is approved by exactly one voter, and  let $k = 1$.

GJCR and GreedyEJR output affordable committees. Another example of a rule ensuring affordability is the variant of sequential CC employed by \citet{EFI+22a}. In light of the positive results on affordable committees in \Cref{sec:guarantees}, it is perhaps not surprising that both rules used by \citet{EFI+22a} satisfy this notion. However, as we will show, affordable rules beyond those considered by \citet{EFI+22a} can be used to obtain strictly stronger guarantees.

Rules returning priceable or affordable committees are among the most popular rules studied in the approval-based committee voting literature, 
as they are able to achieve strong proportionality guarantees with often quite simple (greedy) approaches. However, since priceable rules cannot be exhaustive, they require a completion step. The effects of this completion step are largely unexplored. By providing utilitarian and representation guarantees for the completions of affordable committees, our work contributes to the ongoing discussion on how to make priceable rules exhaustive.%
\footnote{There is variant of priceability that replaces $k$ in \textbf{C2} and \textbf{C5} with an arbitrary budget $B>0$. This weakening of priceability is compatible with exhaustiveness. Examples of exhaustive rules satisfying this property include Phragmén's sequential rule \citep{Phra95a}, leximax-Phragmén \citep{BFJL16a}, the Maximin Support Method \citep{SFFB18a}, and Phragmms \citep{CeSt21a}. Since these rules do not necessarily output affordable committees, our results in \Cref{sec:guarantees} do not apply to them. This explains why, e.g., the representation guarantee of Phragm\'{e}n's sequential rule is only $\frac{1}{2}$ \citep{LaSk20b}, whereas $\frac{3}{4}$ is possible by completing affordable committees. In \Cref{app:price}, we prove that every exhaustive rule satisfying this variant of priceability has a representation guarantee of exactly $\frac{1}{2}$.}

We begin our analysis of affordable committees with the following simple observation, which was already implicit in the work of \citet[Theorem~3.6]{EFI+22a} and \citet[Lemma~1]{PPS21a}. 
\begin{restatable}{observation}{obstriv}
    Let $W \subseteq C$ be an \pricea committee. Then the coverage of $W$ is at least $\frac{\lvert W \rvert}{k}n$.
    \label{obs:triv}
\end{restatable}
\begin{proof}
    If the coverage of $W$ was less than $\frac{\lvert W \rvert}{k}n$, then there must exist a represented voter who pays more than $\frac{\lvert W \rvert }{\frac{\lvert W \rvert}k n} = \frac{k}{n}$, and thus $W$ could not have been \pricea.
\end{proof}
In particular, this implies that an exhaustive affordable committee covers all voters and thus has a coverage of $n$. 
We can also show that any exhaustive affordable committee has a utilitarian guarantee of $\Omega(\frac{1}{k})$. 

\begin{restatable}{observation}{obsexhaust}
    Let $W$ be an exhaustive \pricea committee. Then $W$ has a representation guarantee of $1$ and a utilitarian guarantee of at least $\frac{1}{k}$. 
    \label{prop:exhaustive_price}
\end{restatable}
\begin{proof}
    The representation guarantee immediately follows from \Cref{obs:triv}. For the utilitarian guarantee, we observe that each candidate in $W$ must have at least $\frac{n}{k}$ approvals, while each candidate outside $W$ can have at most~$n$, hence the utilitarian ratio is at least $\frac{n}{nk} = \frac{1}{k}$.
\end{proof}
\begin{sloppypar} While the representation guarantee is optimal, the utilitarian guarantee is far from: \citet{LaSk20b} have shown that proportional rules can achieve a utilitarian guarantee of $\Theta(\frac{1}{\sqrt{k}})$. 
In \Cref{sec:util}, 
we extend \Cref{prop:exhaustive_price} and show that any $f$-representative\,---\,with $f \in \Omega(\ell)$--- \pricea committee containing at least a constant fraction of candidates has a utilitarian guarantee of $\Omega(\frac{1}{\sqrt{k}})$. In particular, this will imply that both MES and GJCR have a utilitarian guarantee of $\Omega(\frac{1}{\sqrt{k}})$ when completed with AV.\footnote{Since MES and GJCR have been proposed relatively recently (and they are not exhaustive), these rules have not been considered in the paper by \citet{LaSk20b}.} \end{sloppypar}

\section{Utilitarian and Representation Guarantees}
\label{sec:guarantees}

We show that optimal utilitarian and representation guarantees can be achieved via completing affordable committees. %

\subsection{Utilitarian Guarantees}
\label{sec:util}

We begin with a general statement, lower bounding the utilitarian ratio of a committee %
based on the highest number of approvals of any unselected candidate. 

For a given committee $W$, let $\uA(W) = \max_{c \notin W} \lvert N_c \rvert $ 
denote the largest approval score of an unselected candidate. 

\begin{restatable}{lemma}{lemlargecom}
    Let $f \colon \mathbb N \to \mathbb R$ and let $W$ be an $f$-representative committee with $\uA(W) \ge t \frac{n}{k}$ for some $t \in \mathbb{N}$.
    Then, $W$ has a utilitarian guarantee of $\min\left(\frac{1}{2}, \frac{f(t)}{2k}\right)$.
    \label{lem:rep_bound}
\end{restatable}
\begin{proof}
    Let $\AVmax$ be the committee chosen by AV and let \mbox{$c \notin W$} such that $\lvert N_c \rvert \ge t \frac{n}{k}$. Since $W$ is $f$-representative, we know that 
    $\sum_{i \in N_c} \lvert A_i \cap W \rvert \ge f(t) \lvert N_c \rvert$. 
    This gives us a lower bound on the utilitarian welfare of $W$: %
    \begin{align*}
        \sw(W) \ge  \max(\sw(W \cap \AVmax), \sum_{i \in N_c} \lvert A_i \cap W \rvert) 
        \ge \frac{\sw(W \cap \AVmax) +  f(t) \lvert N_c \rvert}{2}. 
    \end{align*} 
    Further, since $\sw(AV) \le \sw(W \cap \AVmax) + k \lvert N_c \rvert$, we get 
    \begin{align*}
        \frac{\sw(W)}{\sw(AV)} \ge \frac{\sw(W \cap \AVmax) + f(t) \lvert N_c \rvert}{2(\sw(W \cap \AVmax) + k \lvert N_c \rvert)} \ge \min\left(\frac{1}{2}, \frac{f(t)}{2k}\right).
    \end{align*}
\end{proof}

Hence, if there is an unselected candidate with $\omega( \frac{n}{\sqrt{k}})$ approvals left and our  committee is $\Theta(\ell)$-representative, we have a utilitarian guarantee of $\omega( \frac{1}{\sqrt{k}})$. 

For affordable committees, we can also show that a similar bound can be obtained for the case $\uA(W) \in o(\frac{n}{\sqrt{k}})$ %
(independently of $W$ being $f$-representative).
\begin{restatable}{lemma}{lemsmallgua}
    Let $W$ be an \pricea committee 
    with $\uA(W) \le t \frac{n}{k}$. Then, $W$ has a utilitarian guarantee of $\min\left(\frac{1}{2}, \frac{\lvert W \rvert}{2tk}\right)$. 
    \label{lem:bound_priceable}
\end{restatable}
\begin{proof}
We again let $\AVmax$ denote the committee chosen by AV.
    We know that the utilitarian welfare of $\AVmax$ is at most $\sw(W\cap \AVmax) + k \uA(W)$. On the other hand, the utilitarian welfare of $W$ is at least $\max(\sw(W\cap \AVmax), \lvert W \rvert \frac{n}{k})$, since every candidate in $W$ has at least $\frac{n}{k}$ approvals. Therefore, the utilitarian guarantee of $W$ is at least \begin{align*}
        \frac{\max(\sw(W\cap \AVmax), \lvert W \rvert \frac{n}{k})}{\sw(W\cap \AVmax) + n t} \ge \frac{\sw(W\cap \AVmax) + \lvert W \rvert \frac{n}{k}}{2(\sw(W\cap \AVmax) + n t)}  \ge \min\left(\frac{1}{2}, \frac{\lvert W \rvert}{2tk}\right)
    \end{align*} 
\end{proof}

Thus, in this case, we also obtain a bound better than $\Omega(\frac{1}{\sqrt{k}})$ in case $\uA(W) \in o(\frac{n}{\sqrt{k}})$ and $\frac{\lvert W \rvert}{k} \in \Omega(1)$. Piecing these two lemmas together, we obtain a bound for both cases.

\begin{restatable}{theorem}{reput}
     Any \pricea and  $f$-representative committee $W$
     has a utilitarian guarantee of $\frac{\min\left({\lvert W \rvert}/{\sqrt{k}}, f(\sqrt{k})\right)}{2k}$. 
    \label{prop:exhaustive_price_ejrp}
\end{restatable}
\begin{proof}
    The utilitarian guarantee follows from \Cref{lem:rep_bound,lem:bound_priceable}, 
    depending on whether $\uA(W)$ is at least or at most $\sqrt{k} \; \frac{n}{k}$.
\end{proof}

\noindent As a consequence, since any committee satisfying EJR+ is $\frac{\ell-1}{2}$-representative \citep{BrPe23a}, MES and GJCR both have a utilitarian guarantee of $\Omega(\frac{1}{\sqrt{k}})$ in case they select at least a constant fraction of candidates.
The same proof idea applies to completions of non-exhaustive committees.
\begin{restatable}{corollary}{avut}
     Let $f \in \Omega(\ell)$ and $W$ be an \pricea and $f$-representative committee. Then $\AV{W}$ has a utilitarian guarantee of $\Omega(\frac{1}{\sqrt{k}})$. 
\end{restatable}
\begin{proof}
    If $\lvert W \rvert \ge k - \sqrt{k}$, we get the bound of $\Omega(\frac{1}{\sqrt{k}})$ from \Cref{prop:exhaustive_price_ejrp}. 
    Otherwise, $\AV{W}$ would select at least the top $\sqrt{k}$ candidates with the highest approval score and therefore achieve a utilitarian guarantee of $\frac{1}{\sqrt{k}}$.
\end{proof}

For committees satisfying EJR+, we give a precise bound.

\begin{restatable}{corollary}{corejrpbound}
     Let $W$ be an \pricea committee that satisfies EJR+. Then, $\AV{W}$ has a utilitarian guarantee of $\frac{1}{4\sqrt{k}} - \frac{1}{2k}$. 
\end{restatable}
\begin{proof}
    If $\lvert W \rvert \ge \frac{3}{4}k$ we can apply \Cref{lem:bound_priceable} and \Cref{lem:rep_bound} and get a guarantee of $\frac{\sqrt{k} - 2}{4k} = \frac{1}{4 \sqrt{k}} - \frac{1}{2k}$, since every EJR+ committee is $\frac{\ell - 1}{2}$-representative. If, on the other hand, $\lvert W \rvert \le \frac{3}{4}k $ we would select at least the $\frac{1}{4}k - 1$ candidates with the highest approval score and thus achieve a utilitarian guarantee of $\frac{1}{4} - \frac{1}{k}$.
\end{proof}

Interestingly, these strong guarantees for affordable EJR+ committees do not hold for EJR. 
\begin{restatable}{theorem}{ejrimpo}
    For any $0 < c < 1$, there are exhaustive \pricea committees satisfying EJR with a utilitarian ratio of $\mathcal{O}(\frac{1}{k ^ {1-c}})$.
\end{restatable}
\begin{proof}
    For simplicity, we assume that $k$ is even, that 
    $k^c$ and $k^{1-c}$
are integers, and that $k$ is sufficiently large. 

    We create an instance with $2k$ candidates $C = C_1 \cup C_2$, with both $C_1$ and $C_2$ being of size $k$. For each subset $C' \subseteq C_1$ of size exactly $|C'|=\frac{k}{2}$, we add a voter approving all candidates in $C'$.  Thus, the number of voters is \[
    n = {k \choose \frac{k}{2}} = \frac{k!}{(\frac{k}{2})! (k - \frac{k}{2})!}.
    \] %
    
    Further, we arbitrarily partition the voters, and the candidates in $C_2$, into $k^{1-c}$ equal-size groups, with the $\frac{n}{k^{1-c}}$ voters in one group approving exactly the $k^{c}$ candidates approved by no other group. %
    
     We show that this instance does not contain an $\ell$-large and $\ell$-cohesive group for any $\ell$ larger than $k^c$. Let $C' \subseteq C$ be a subset of candidates of size larger than $k^c$. If $C' \cap C_2 \neq \emptyset$, we know that there are at most $\frac{n}{k^{1-c}}$ voters approving $C' \cap C_2$ and thus the group is not large enough to be more than $k^c$-large. Thus, we know that $C' \subseteq C_1$. There are at most  
    \begin{align*}
    {(k - \lvert C' \rvert) \choose (\frac{k}{2} - \lvert C' \rvert)} < {(k - k^c) \choose (\frac{k}{2} - k^c)} = \frac{(k - k^c)! }{(\frac{k}{2} - k^c)! (k - \frac{k}{2})!}
    \end{align*} voters approving all of $C'$. By definition of the binomial coefficient, we get that the fraction of voters approving all of $C'$ is at most 
    \begin{align*}
        \frac{{(k - k^c) \choose (\frac{k}{2} - k^c)}}{{k \choose \frac{k}{2}}} = \frac{(k - k^c)!\frac{k}{2}! }{(\frac{k}{2} - k^c)!k!} = \frac{\prod_{i = \frac{k}{2} - k^c + 1}^{k - k^c} i }{\prod_{i = \frac{k}{2} + 1}^{k} i}   = \frac{\prod_{i = \frac{k}{2} - k^c + 1}^{\frac{k}{2}} i }{\prod_{i = k - k^c + 1}^{k} i} \le \left(\frac{(\frac{k}{2})}{k - k^c} \right)^{k^c}. 
    \end{align*}
    Since $k$ is sufficiently large, we can assume that this is smaller than $\frac{k^c}{k}$ and, thus, no $k^c$-cohesive group exists.
    Hence, we can take the other $k$ candidates in $C_2$ and get a committee which is \pricea and satisfies EJR. However, their utilitarian welfare is only $n k^c$ while the utilitarian welfare of the other $k$ candidates is $n \frac{k}{2}$. Thus, no utilitarian ratio better than $O(\frac{1}{k^{1-c}})$ is possible for all \pricea committees satisfying EJR.
\end{proof}

We point out that the committee witnessing the utilitarian ratio of $\mathcal{O}(\frac{1}{k ^ {1-c}})$ can be selected by the GreedyEJR rule, since no cohesive groups with an $\ell$ larger than $k^{c}$ exists. Thus, the GreedyEJR rule, even when it is exhaustive (and thus cannot be completed any further), can select committees with a very suboptimal utilitarian ratio.

We now show that a utilitarian guarantee of $\frac{2}{\sqrt{k}} - \frac{1}{k}$ can be achieved via completing the Greedy Justified Candidate Rule with AV, thus meeting the lower bound proven by \citet{LaSk20b} and answering an open question by \citet{EFI+22a}. For our proof, we use that GJCR behaves very similarly to AV at the beginning of its iterations, picking the candidates with the highest approval score, until one of them would not be ``proportional'' anymore. 
\begin{theorem}
    The Greedy Justified Candidate Rule completed with AV has a utilitarian guarantee of $\frac{2}{\sqrt{k}} - \frac{1}{k}$.
    \label{thrm:opt_bound}
\end{theorem}
\begin{proof}
    Fix an instance and let $W$ be the committee output by GJCR. 
    Let $c_1, \dots, c_k$ be the candidates selected by AV in order of their approval scores. Further, let $c_{i+1}$ be the first among them which was not selected by $\AV{W}$. Without loss of generality, we assume that all other candidates selected by $\AV{W}$ with the same approval score as $c_{i+1}$ come before it.  Let $\alpha > 0$ such that $\lvert N_{c_{i+1}} \rvert = \alpha \frac{n}{k}$. 
    We distinguish two cases. If $\alpha < 1$, we have {$\lvert N_{c_{i+1}} \rvert < \frac{n}{k}$}. 
    Thus, $c_{i+1}$ could not have been selected by GJCR, since all candidates selected by it have an approval score of at least~$\frac{n}{k}$. Hence, $\AV{W}$ and AV select candidates with the same approval scores, and therefore the utilitarian guarantee is $1$.

    For the case of $\alpha \ge 1$, we first want to show that \mbox{$i \ge \lceil \alpha \rceil$}.
    Since the candidate $c_{i+1}$ was not selected, we know that there is at least one voter $j \in N_{c_{i+1}}$ with $\lvert A_j \cap W \rvert \ge \lfloor \alpha \rfloor$ such that the candidates selected for $j$ (meaning that $j \in N'$ when this candidate was chosen by GJCR) must have at least as many approvals as $c_{i+1}$. 
    Therefore, $i \ge \lfloor \alpha \rfloor$ immediately holds (and thus also $i \ge \lceil \alpha \rceil$ if $\alpha \in \mathbb{N}$ or if there are at least $\lceil \alpha \rceil$ candidates selected for $j$). If $\alpha \notin \mathbb{N}$ and if there are less than $\lceil \alpha \rceil$ selected for $j$, we know that in the payment system constructed by GJCR, the budget of $j$ and thus also the entire budget cannot be fully spent, since the voter pays at most $\frac{\lfloor \alpha \rfloor k}{\alpha n} < \frac{k}{n}$. Therefore, the approval voting step selects at least one candidate, which must have at least an approval score of $\alpha \frac{n}{k}$. Hence, in this case, $i \ge \lfloor \alpha \rfloor + 1 = \lceil \alpha \rceil$. Finally, we know that the other $k-i$ candidates selected by GJCR must have at least $\frac{n}{k}$ approvals since they were selected by GJCR, while we know that the candidates selected by AV afterwards must have an approval score of at most $\alpha \frac{n}{k}$.
    
    We can therefore lower-bound the utilitarian guarantee as 
    \begin{align*}
       & \frac{\sw(\AV{W}))}{\sw(AV)} \ge \frac{\sw(\{c_1, \dots, c_i)\} + (k - i) \frac{n}{k}}{\sw(\{c_1, \dots, c_i\}) +( k - i) \alpha\frac{n}{k}  } 
        \ge \frac{i\alpha   \frac{n}{k} + (k - i) \frac{n}{k}}{i\alpha   \frac{n}{k} + ( k -i) \alpha\frac{n}{k}}\\ & = \frac{\alpha i+ (k - i)} {\alpha i+ ( k - i) \alpha} = \frac{\alpha i+ (k - i)}{k \alpha} = \frac{(\alpha - 1) i + k}{k \alpha} \ge \frac{(\alpha - 1) \alpha + k}{k \alpha} = \frac{\alpha^2 + k}{k \alpha} - \frac{1}{k}.
    \end{align*} 
    By the AM-GM inequality we know that $(\alpha^2 + k)/2 \ge \alpha \sqrt{k}$ and thus we obtain a utilitarian guarantee of $ \frac{2}{\sqrt{k}} - \frac{1}{k}$.
\end{proof}

In the same vein, we can obtain a very similar bound for MES completed by AV. %
\begin{restatable}{theorem}{mesutibound}
    The Method of Equal Shares completed by AV has a utilitarian guarantee of $\frac{2}{\sqrt{k}} - \frac{2}{k}$.
    \label{thrm:opt_bound_mes}
\end{restatable}
\begin{proof}
The proof is very similar to the one for GJCR. However, for MES we can only claim that $i \ge \lfloor \alpha \rfloor$. Any candidate who has less approvals than $\alpha \frac{n}{k}$ could only be bought with a $\rho$ of larger than $\frac{k}{n\alpha}$ while any candidate in $c_1, \dots, c_i$ could be bought with a $\rho$ of at most $\frac{k}{n\alpha}$.  Since $c_{i+1}$ was not selected for a $\rho$ of $\frac{k}{n\alpha}$, we know that there must be an approver of $c_{i+1}$ with a budget of less than $\frac{k}{n\alpha}$. Since this approver paid at most $\frac{k}{n\alpha}$ for each previous candidate, we get that this voter must approve at least $\lfloor \alpha\rfloor$ candidates with an approval score of at least $\frac{\alpha n}{k}$.
    
The rest follows analogously, and we can give the utilitarian guarantee as 
    \begin{align*}
        &\frac{\sw(\AV{W}))}{\sw(AV)} \ge \frac{\sw(\{c_1, \dots, c_i\}) + (k - i) \frac{n}{k}}{\sw(\{c_1, \dots, c_i\}) +( k - i) \alpha\frac{n}{k}  } 
       \ge \frac{i\alpha   \frac{n}{k} + (k - i) \frac{n}{k}}{i\alpha   \frac{n}{k} + ( k -i) \alpha\frac{n}{k}} \\ & = \frac{(\alpha - 1) i + k}{k \alpha} \ge \frac{(\alpha - 1) \lfloor \alpha \rfloor + k}{k \alpha} \ge \frac{(\alpha - 1)^2 + k}{k \alpha} \ge \frac{\alpha^2 + k}{k \alpha} - \frac{2}{k}.
    \end{align*}
    Hence, we can again apply the AM-GM inequality and obtain a utilitarian guarantee of $ \frac{2}{\sqrt{k}} - \frac{2}{k}$.
\end{proof}
\subsection{Representation Guarantees}
We begin the analysis of representation guarantees with a simple lemma: Generalizing the proof of Theorem~3.6 by \citet{EFI+22a} gives us a general bound on the representation guarantee of $\CC{W}$. 
\begin{restatable}{lemma}{boundrepprice}
    Let $W$ be a committee with a representation ratio of~$\rho$. Then, $\CC{W}$ has a representation guarantee of $\rho + (1 - \rho) (\frac{k - \lvert W \rvert}{k})$.
    \label{lem:boundrepprice}
\end{restatable}
\begin{proof}
    Let $n'$ be the number of voters covered by the optimal CC committee. 
    Since $W$ has a representation ratio of $\rho$ there are at least %
    $(1 - \rho)n'$ of these voters uncovered. Since at least $(1 - \rho)n'$ uncovered voters can be covered using $k$ candidates, using $k - \lvert W \rvert$ candidates we can cover at least $\left(\frac{k - \lvert W \rvert}{k}\right) (1-\rho)n'$. The representation guarantee thus follows.
\end{proof}

This allows us to derive the same representation guarantee as \citet{EFI+22a}, but for any affordable committee completed with~CC.
\begin{restatable}{corollary}{coroptrep}
    Let $W$ be an \pricea committee. Then $\CC{W}$ has a representation guarantee of $\frac{3}{4}$.
    \label{cor:threequart}
\end{restatable}
\begin{proof}
    Let $\rho$ be the representation ratio of $W$. By \Cref{obs:triv}, we know that $\rho \ge \frac{\lvert W \rvert}{k}$. Using \Cref{lem:boundrepprice}, we get that the representation guarantee is at least 
    \begin{align*}
    \rho + (1 - \rho)\left(\frac{k - \lvert W \rvert}{k}\right) = \left(1 - \frac{\lvert W \rvert}{k}\right) + \rho \frac{\lvert W \rvert}{k} \ge \left(1 - \frac{\lvert W \rvert}{k}\right) + \left(\frac{\lvert W \rvert}{k}\right)^2.
    \end{align*} To bound this, we notice that $(1-x) + x^2$ for $x \in [0,1]$ is minimized for $x = \frac{1}{2}$ with $\frac{1}{2} + \frac{1}{2^2} = \frac{3}{4}$.
\end{proof}

As a consequence, the proportionality notions EJR+ and FJR \citep{PPS21a} \footnote{FJR or fully justified representation is a strengthening of EJR, which relaxes the cohesiveness requirement and thereby strengthens the axiom. It is still satisfiable via a generalization of the GreedyEJR rule.} are compatible with a representation guarantee of $\frac{3}{4}$. 
\begin{restatable}{corollary}{repfjr}
    For every instance, there exist committees $W$ and $W'$, each with representation ratio of at least $\frac{3}{4}$, such that $W$ satisfies EJR+ and $W'$ satisfies FJR.
\end{restatable}
\begin{proof}

    This follows from the fact that there are \pricea committees satisfying EJR+ (for instance,  the output of MES or GJCR) or satisfying FJR (for instance, the output of the \textit{greedy cohesive rule} of \citealp{PPS21a}).  
\end{proof}
Next, we turn to completing committees via the \textit{sequential} Chamberlin--Courant (seq-CC) rule. It is well known that seq-CC achieves a representation guarantee of $1-\frac{1}{e}\approx 0.632$ and that this is the best possible for any polynomial-time rule (assuming P $\neq$ NP) \citep{Feig98a}. Perhaps surprisingly, we show that the analysis of seq-CC does not depend on the first steps being taken optimally, but that \textit{any} affordable committee, completed with seq-CC, has a representation guarantee of $1 - \frac{1}{e}$. Thus, additionally imposing strong proportionality notions (e.g., EJR+) does not lower the best possible representation guarantee of a polynomial-time computable rule.
\begin{restatable}{theorem}{seqcc}
    Let $W$ be an \pricea committee. Then, 
    $\overline{\text{seq-CC}}(W)$
    has a representation guarantee of $1-\frac{1}{e}$. 
    \label{thmseqcc}
\end{restatable}
\begin{proof}
    First, by \Cref{obs:triv}, the committee $W$ has a representation ratio of at least $\frac{\lvert W\rvert}{k}$. Let $\alpha = \cov\left(CC(\mathcal{I})\right)$ be the optimal coverage. Now, seq-CC iteratively adds a candidate $c \notin W$ to $W$ such that $\cov(W \cup \{c\})$ is maximized. Let $r$ be the representation ratio in a step of seq-CC. Then there are at least $(1-r)\alpha$ of the voters covered by $CC(\mathcal{I})$ uncovered. Since they can be covered using $k$ candidates, there is at least one candidate covering $\frac{(1-r)\alpha}{k}$ of these voters. Thus, after adding the candidate to $W$, the total number of uncovered voters is at most $$(1 - r)\alpha - (1-r) \alpha\frac{1}{k} = (1-r)\left(1-\frac{1}{k}\right)\alpha$$. By induction, we get that after $k - \lvert W \rvert$ steps, there are at most $$\left(1 - \frac{\lvert W\rvert}{k}\right) \left(1-\frac{1}{k}\right)^{k-\lvert W \rvert} \alpha \le \left(1 - \frac{1}{k}\right)^k \alpha \le \frac{\alpha}{e}$$ uncovered voters. Since $\alpha$ is the optimal representation ratio, we get that the representation guarantee of our approximation is at least $1- \frac{1}{e}$.
\end{proof}

\subsection{Combining Utilitarian and Representation Guarantees}

Finally, we turn to combining both types of guarantees. For both ratios, we show that they can be approximated optimally, while still maintaining a close-to-optimal approximation to the other ratio.   We first show that a representation guarantee of exactly $\frac{3}{4}$ is compatible with an asymptotically optimal utilitarian guarantee and EJR+.
\begin{restatable}{theorem}{thmbalanceone}
    For any $c > 2$, there exists a committee satisfying EJR+ with a representation guarantee of $\frac{3}{4}$ and a utilitarian guarantee of $\min\left(\frac{2}{\sqrt{ck}} - \frac{2}{k}, \frac{(c-2)^2}{4c^2} - \frac{1}{k}\right) \in \Omega(\frac{1}{\sqrt{k}})$.
    \label{thm:balance1}
\end{restatable}
\begin{proof}
    Let $W$ be the committee selected by GJCR. We distinguish two cases. 
    
    \textbf{Case 1:} $\lvert W \rvert \ge \frac{k}{c}$. In this case, we complete $W$ entirely with CC. As shown in \Cref{thrm:opt_bound}, this gives us a representation guarantee of $\frac{3}{4}$. 
    
     For the utilitarian guarantee,  we follow the proof of \Cref{thrm:opt_bound}. We again sort the candidates $c_1, \dots, c_m$ by their approval scores and let $c_{i+1}$ be the candidate with the largest approval score $\alpha \frac{n}{k}$ not selected by GJCR. Based on the way GJCR is run, we know that $i \ge \lfloor \alpha \rfloor$ since $c_{i+1}$ is not included in the committee. For the utilitarian guarantee, we ignore the candidates added by CC and lower bound the guarantee by  
    \begin{align*}
        &\frac{\sw\left(\{c_1, \dots, c_i\}\right) + \left(\frac{k}{c} - i\right) \frac{n}{k} }{\sw(\{c_1, \dots, c_i\}) + (k - i) \alpha \frac{n}{k}} \ge \frac{i \alpha \frac{n}{k} + \left(\frac{k}{c} - i\right) \frac{n}{k}}{i \alpha \frac{n}{k} + ( k -i) \alpha\frac{n}{k}} = \frac{i \alpha+ (\frac{k}{c} - i)} {i \alpha+ ( k - i) \alpha} \\ &= \frac{(\alpha - 1)i + \frac{k}{c} }{k \alpha} \ge\frac{(\alpha- 1)^2 + \frac{k}{c}}{k \alpha} \ge \frac{\alpha^2 + \frac{k}{c}}{k\alpha} - \frac{2}{k}  
        \ge \frac{2\alpha \frac{\sqrt{k}}{\sqrt{c}}}{k\alpha } - \frac{2}{k} = \frac{2}{\sqrt{ck}} - \frac{2}{k} \; \text,
    \end{align*}
    where in the last row we applied the AM-GM inequality.

    \smallskip
    \textbf{Case 2:} $\lvert W \rvert \le \frac{k}{c}$. In this case, we can select the $x \coloneqq \left\lfloor \frac{(c-2)^2}{4 c^2} k \right\rfloor$ candidates with the highest approval scores, thus guaranteeing a utilitarian guarantee of at least $\frac{(c-2)^2}{4c^2} - \frac{1}{k}$. We note that \begin{align*}\frac{k}{c} + \frac{(c-2)^2}{4 c^2} k = k\frac{c^2 + 4}{4c^2} \le k \frac{4c^2 - 4c}{4c^2}= k - \frac{k}{c} \; \text, \end{align*}
    where for the inequality we used that $c \ge 2$. Thus, we can indeed select $x$ candidates without using more than $k$ candidates in total.

    For the representation guarantee, we observe that the function \mbox{$x + (1-x)^2$} is monotonically decreasing between $0$ and $\frac{1}{2}$. Thus, using \Cref{lem:boundrepprice}, we get a representation guarantee of at least \begin{align*}
       & \frac{\lvert W \rvert}{k} + \left(1 - \frac{\lvert W \rvert}{k}\right)\left(1 - \frac{\lvert W \rvert + x}{k}\right) = \frac{\lvert W \rvert}{k} + \left(1 - \frac{\lvert W \rvert}{k}\right)^2 - \left(1 - \frac{\lvert W \rvert}{k}\right) \frac{x}{k} \\&\ge \frac{1}{c} + \left(1 - \frac{1}{c}\right)^2 - \left(1 - \frac{\lvert W \rvert}{k}\right) \frac{x}{k} \ge 1 + \frac{1}{c^2}  - \frac{1}{c} - \frac{x}{k} = \frac{3}{4} + \frac{(c-2)^2}{4c^2} - \frac{x}{k} \ge \frac{3}{4}.
    \end{align*}
\end{proof}

Moreover, we show that one can achieve a representation guarantee of $\frac{3}{4} - \mathcal{O}(\frac{1}{\sqrt{k}})$ and a utilitarian guarantee of $\frac{2}{\sqrt{k}} - \frac{1}{k}$ simultaneously with EJR+, thereby strengthening 
the result of \citet[Theorem~3.3]{EFI+22a} with a simpler proof.

\begin{restatable}{theorem}{thmbalancetwo}
There exists a committee satisfying EJR+ with a representation guarantee of $\frac{3}{4} - \frac{2}{\sqrt{k}}$ and a utilitarian guarantee of $\frac{2}{\sqrt{k}} - \frac{1}{k}$. \end{restatable}
\begin{proof}
    Again, let $W$ be the output of GJCR. If $\lvert W \rvert \ge \frac{3}{4}k$, then by \Cref{obs:triv} we know that $W$ must have a representation guarantee of at least $\frac{3}{4}$. Thus, by \Cref{thrm:opt_bound}, the committee $\AV{W}$ achieves both bounds. 

    If $\lvert W \rvert \le \frac{3}{4}k$, we first complete it with the $\lfloor 2 \sqrt{k}\rfloor$ candidates with the highest approval score and then with CC. This, by definition, achieves a utilitarian guarantee of $\frac{2}{\sqrt{k}} - \frac{1}{k}$. Let $W'$ be the committee $W$ together with the $\lfloor 2 \sqrt{k}\rfloor$ candidates and let $\rho$ be the representation guarantee of $W'$. By \Cref{lem:boundrepprice}, we get that the representation guarantee of $\CC{W'}$ is at least
    \begin{align*}\rho &+ (1-\rho)\frac{k - \lvert W \rvert - \lfloor 2\sqrt{k}\rfloor}{k}
    \ge \frac{\lvert W \rvert}{k} + \left(\frac{k - \lvert W \rvert}{k} \frac{k - \lvert W \rvert - 2\sqrt{k}}{k}\right) \\ 
    &= \frac{\lvert W \rvert}{k} + \left(\frac{k - \lvert W \rvert}{k}\right)^2 - \left(\frac{k - \lvert W \rvert}{k}\right)\frac{2}{\sqrt{k}} \ge \frac{3}{4} - \frac{2}{\sqrt{k}}.
    \end{align*}
\end{proof}

\section{Other Completion Methods}
\label{sec:other}

There are other ways to complete committees besides AV and (seq-)CC. In this section, we discuss three of them.

\subsection{Maximin Support}

The \textit{maximin support} of a committee is a quantitative notion measuring the overrepresentation of voters. \citet{CeSt21a} argued that this measure is relevant for the security of a blockchain consensus mechanism known as \textit{Nominated Proof-of-Stake}. 
They also showed that the {maximin support method}, a voting rule introduced by \citet{SFFB18a},  approximates this measure. 

To formally define these concepts, we adopt the formulation of \citet{CeSt21a}. The maximin support $\supp(W)$ of a committee $W$ is given by 
\(
\min_{S \subseteq W, S \neq \emptyset} \frac{1}{\lvert S \rvert} \lvert \{i \in N \colon A_i \cap S \neq \emptyset \}\rvert 
\), %
and the goal is to maximize this value.
The \textit{maximin support method (MMS)} starts with an empty committee $W$ and iteratively adds a candidate $c \notin W$ maximizing $\supp(W \cup \{c\})$. This rule achieves a $\frac{1}{2}$-approximation to the maximin support objective \citep{CeSt21a}, but does not satisfy EJR \citep{SFFB18a}.

First, we give an example of an instance in which every committee satisfying EJR approximates the optimal maximin support value by a factor of at most $\frac{2}{3}$. 
\begin{example}
    Consider an instance with $2k$ voters and $2k$ candidates $c_1, \dots, c_{2k}$. Voters $i = 1, \dots, k$ each only approve candidate $c_i$, while voters $i = k+1, \dots, 2k$ approve all of $c_{k+1}, \dots, c_{2k}$ and $c_{i - k}$. Now, the optimal committee with regard to the maximin support would include the candidates $c_1, \dots, c_k$ with an optimal maximin support of $\frac{n}{k} = 2$. On the other hand, to satisfy EJR we need to include at least $\frac{k}{2}-1$ of the candidates $c_{k+1}, \dots, c_{2k}$. Hence, there must be at least $\frac{k}{2}-1$ voters from $1$ to $k$ without a candidate they approve in the committee. The maximin support is thus at least $\frac{1}{k}(k + \frac{k}{2} + 1) = 1.5 + \frac{1}{k}$ and therefore a maximin support approximation of $\frac{2}{3} + \varepsilon$ is not possible.
\end{example}

Second, we show that a simple application of a lemma from \citet{CeSt21a} shows that any affordable committee, completed with the maximin support rule, achieves an approximation to the maximin support of $\frac{1}{2}$. 

\begin{restatable}{theorem}{maximin}
    Let $W$ be an \pricea committee. Then the committee obtained by completing $W$ by running the maximin support rule approximates the maximin support objective by a factor of $\frac{1}{2}$.
\end{restatable}
\begin{proof}
    Let $x$ be the optimal obtainable maximin support value. As shown by \citet[Lemma~10]{CeSt21a}, for any committee $W$ there always exists a candidate $c \notin W$ with $\supp(W \cup \{c\}) \ge \min\left(\supp(W), \frac{1}{2}x\right)$.

    First,  since the optimal committee $W'$ has size \mbox{$|W'|=k$}, we observe that $$\frac{1}{\lvert W'\rvert} \left\lvert \{i \in N \colon A_i \cap W' \neq \emptyset\}\right\rvert \le \frac{n}{k}$$ and thus $x \le \frac{n}{k}$. Further, since the committee is affordable, we know that for any subset $S \subseteq W$ it has to hold that $\lvert \{i \in N \colon A_i \cap S \neq \emptyset\} \rvert \ge \frac{\lvert S \rvert n}{k}$ and thus $\supp(W) \ge \frac{n}{k} \ge x$. Hence, inductively by adding candidates via the maximin support rule, it holds that in each step the maximin support is at least $\frac{1}{2}x$ and, therefore, we have a $\frac{1}{2}$-approximation.
\end{proof}

This approximation factor is currently the best known polynomial-time approximation to the maximin support. Interesting question for future work include (i)~whether there is a polynomial-time rule with a better approximation factor than~$\frac 1 2$, and (ii) whether there always exists a committee satisfying EJR and providing an approximation to the maximin support value of better than $\frac 1 2$. 
\subsection{Other Completions}

Next, we consider two completion methods that have been suggested in the PB literature \citep{ReMa23a} for the Method of Equal Shares: ``completion by varying the budget'' and ``completion by perturbation.'' 

\paragraph{Completion by Varying the Budget} %
This completion method finds the first committee size $k' \ge k$ for which MES would select at least $k$ candidates. If, for committee size $k'$, MES selects a committee of size exactly $k$, then this committee is output. Otherwise, the completion method outputs the committee selected for committee size $k'-1$. This, however, does not need not be exhaustive (or even select any candidates) as the following simple example shows.%
\footnote{We note that, due to the non-monotonicity of MES, this method does not necessarily select a superset of the committee selected by MES \citep{ReMa23a}, and is thus not a completion method in the technical sense.} 
\begin{example}
    Consider an instance with $n = 6, k = 2$ and three candidates $c_1, c_2, c_3$ such that each candidate is approved by exactly two voters. For $k = 2$, MES would output nothing, while for $k = 3$ it would choose all candidates. 
    \label{example:varying}
\end{example}
This shows that a second completion step is necessary. \citet{BFJK23a} and \citet{FFP+23a} suggest to use (a PB analog of) AV for this second step. 
We show that this has the same guarantees as MES completed with~AV.

\begin{restatable}{theorem}{mesvarying}
    MES completed by varying the budget and AV achieves a representation guarantee of $\frac{1}{k}$ and a utilitarian guarantee of~$\frac{2}{\sqrt{k}} - \frac{2}{k}$. Both bounds are asymptotically tight.
\end{restatable}
\begin{proof}
    Let $B \ge k$ the budget for which MES is finally run and let $W$ be the selected committee. 

    For the representation guarantee, it is easy to see that the candidate with the most approvals. which is guaranteed to be selected, covers at least $\frac{1}{k}$ of the voters. 

    For the utilitarian guarantee, this follows from a similar case distinction as in \Cref{thrm:opt_bound} and \Cref{thrm:opt_bound_mes}. For the unselected candidate with the largest approval score $c_{i + 1}$ (with $\lvert N_{c_{i+1}} \rvert = \alpha \frac{n}{k})$, MES must buy at least $\left\lfloor \alpha \frac{B}{k} \right\rfloor$ candidates approved by at least $\alpha \frac{n}{k}$ voters. Further, each candidate selected afterwards has at least an approval score of $\frac{n}{B}$. 
    Thus, we get a utilitarian guarantee of 
    \begin{align*}
       & \frac{\sw(\AV{W}))}{\sw(AV)} \ge \frac{\sw(\{c_1, \dots, c_i\}) + (k - i) \frac{n}{B}}{\sw(\{c_1, \dots, c_i\}) +( k - i) \alpha\frac{n}{k}  } 
        \ge \frac{i\alpha   \frac{n}{k} + (k - i) \frac{n}{B}}{i\alpha   \frac{n}{k} + ( k -i) \alpha\frac{n}{k}} = \frac{\alpha i + (k - i) \frac{k}{B}}{k\alpha} \\ & = \frac{(\alpha - \frac{k}{B})i + \frac{k^2}{B}}{k\alpha} \ge \frac{(\alpha - \frac{k}{B})\lfloor \alpha \frac{B}{k} \rfloor + \frac{k^2}{B}}{k\alpha} \ge \frac{(\alpha - \frac{k}{B})(\alpha \frac{B}{k} - 1) + \frac{k^2}{B}}{k\alpha} \\ &\ge  \frac{\alpha^2\frac{B}{k} - 2\alpha + \frac{k^2}{B}}{k\alpha} = \frac{\alpha^2\frac{B}{k}  + \frac{k^2}{B}}{k\alpha} - \frac{2}{k} \ge \frac{2 \sqrt{\frac{\alpha^2B}{k}}\sqrt{\frac{k^2}{B}}}{k\alpha} - \frac{2}{k} = \frac{2}{\sqrt{k}} - \frac{2}{k}.
    \end{align*}

    The asymptotic tightness of the utilitarian guarantee follows from the general upper bound on rules satisfying JR \citep{LaSk20b, EFI+22a}. For the representation guarantee, consider the instance in \Cref{example:varying} with one of the candidates cloned $k$ times. This candidate and its clones could be selected, leading to a coverage of $\frac{n}{k+1}$, while a coverage of $\frac{k n}{k+1}$ is possible.
\end{proof}

\paragraph{Completion by Perturbation}
This variant, suggested by \citet{PPS21a}, is also known as the ``epsilon method.'' 
We give it here in the formulation of \citet{FFP+23a}. When MES terminates, the approvers of each unselected candidate $c \notin W$ do not have enough budget left to fund $c$. The completion method calculates, for each $c \notin W$, the $\rho$ such that $\sum_{i \in N_c} b_i + \sum_{i \notin N_c} \min(b_i, \rho) = 1$. It selects the candidate minimizing $\rho$ and updates the budgets accordingly. Thus, voters \textit{not} approving the candidate help fund it, but should spend as little as possible. 
It is easy to see that this completion method makes MES exhaustive. We show that it achieves a representation guarantee of $\frac{1}{2}$, but a utilitarian guarantee of only~$\frac{1}{k}$. %

\begin{restatable}{theorem}{mespert}
    MES completed by perturbation has a utilitarian guarantee of $\Theta(\frac{1}{k})$ and a representation guarantee of~$\frac{1}{2}$. Both bounds are asymptotically tight.
\end{restatable}
\begin{proof}
    For the representation guarantee, let $W$ be the selected committee. Further, assume that there is another committee $W'$ of size $k$ covering more than twice as many voters as $W$. Let $N'$ be the voters covered by $W$ and $N''$ the voters covered by $W'$ but not by $W$. Since $W'$ covers more than twice as many voters as $W$ we get that $\lvert N'' \rvert \ge \lvert N' \rvert$.
    
    First, we notice that if an originally unaffordable candidate is bought, the voters approving this candidate must spend their whole budget. Let $N_{1}, \dots, N_k$ be a partition of $N'$ such that the last candidate (during the execution of MES) on which the voters in $N_i$ spend their spent their budget is candidate $c_i$. By the pigeon-hole principle, there must a $N_i$ with $\lvert N_i \rvert < \lvert W' \setminus W \rvert /k$. However, when this candidate was bought, instead the candidate from $W' \setminus W$ covering the largest set of approvers no one in $W$ could have been selected, achieving a lower $\rho$ since the voters approving have not spent any budget on candidates they approve. (We note that this proof is very similar to the proof of \Cref{thm:app:pricerep}. However, MES completed by perturbation does not need to be priceable.)

    The utilitarian guarantee of $\frac{1}{k}$, on the other hand, is the trivial guarantee which follows from the fact that MES completed by perturbation selects the approval winner in the first step.
    
    The asymptotic tightness of the representation guarantee follows immediately, since it is a constant-factor approximation. For the utilitarian guarantee, consider an instance consisting of $\frac{n}{2k}$ voters approving $k$ candidates together and $n - \frac{n}{2k}$ voters, each approving a single candidate only approved by them. MES completed by perturbation would first select one candidate approved by the $\frac{n}{2k}$ voters, deplete their budget and then select $k-1$ candidates arbitrarily, This has a utilitarian welfare of $\frac{n}{2k} + k-1$ while the optimal utilitarian welfare is $\frac{n}{2}$ thus showing a utilitarian guarantee of $\mathcal{O}(\frac{1}{k})$.
\end{proof}

\section{Conclusion and Future Work}

We studied trade-offs between proportionality,  coverage (aka diversity), and utilitarian welfare (aka individual excellence) in the setting of approval-based multiwinner voting. We showed that very good compromises, even under strong proportionality notions, can be achieved by completing so-called \textit{affordable} committees. Rules that always output affordable committees include the Method of Equal Shares \citep{PeSk20b} and the Greedy Justified Candidate Rule \citep{BrPe23a}. For the latter, we showed that it can be completed to achieve an optimal utilitarian guarantee of $\frac{2}{\sqrt{k}} - \frac{1}{k}$, thereby answering an open question of \citet{EFI+22a}. For diversity, we showed that any affordable committee can be completed to satisfy a representation guarantee of $\frac{3}{4}$. We also studied trade-offs between both guarantees and showed that they can be approximated close to optimal simultaneously, while still satisfying EJR+.

We highlight two directions for future work that go beyond approval-based multiwinner voting. 
For elections based on ordinal (rank-order) preferences, trade-offs between (the ordinal analogue of) the Chamberlin--Courant rule and the Borda rule, as well as the Monroe rule, have been studied \citep{KKE+19a,FaTa18a}. It would be interesting to extend this line of research by analyzing the ``price of fairness'' for axioms such as \textit{proportionality for solid coalitions} \citep{Dumm84a} or the recently introduced notions of \citet{BrPe23a}. Furthermore, ordinal rules such as STV \citep{Tide95a} or the Expanding Approvals Rule \citep{AzLe20a} may require completion methods as well, in the case of truncated, i.e., incomplete, ballots. 

Moreover, it would be interesting to revisit the participatory budgeting setting. Although the results by \citet{FVMG22a} were mostly negative, there remain several possibilities for more positive results. 
As a first example, we note that \citet{FVMG22a} exclusively look at the utility function that counts the approved projects in the outcome. This, however, is not the only way to measure utility. For instance, in real-world applications of MES, the utility of a voter is often measured by the \textit{cost} of approved projects in the outcome (rather than their number); and several other utility functions have been discussed in the PB literature \citep{BFL+23a, MREL23a}. It is, therefore, natural to ask whether the impossibility results of \citet{FVMG22a} extend to these functions, or whether positive results are possible. 

As a second example, 
a lot of axiomatic work has focused on ``relaxed'' notions of fairness, such as \textit{EJR up to one project} \citep{PPS21a}.
An intriguing question is whether meaningful utilitarian and representation guarantees \emph{up to one project} can be achieved.

\subsection*{Acknowledgments}
This research is supported by the Deutsche Forschungsgemeinschaft (DFG) under the grant BR~4744/2\nobreakdash-1 and the Graduiertenkolleg ``Facets of Complexity'' (GRK~2434). We thank Piotr Faliszewski for helpful feedback.

\bibliographystyle{plainnat}
\bibliography{abb, dist, algo_fork}

\clearpage
\appendix

\section{Priceability with Variable Budget}
\label{app:price}

Consider an instance $\mathcal{I} = (A, C, k)$.  
A committee $W\subseteq C$ is said to be \emph{$B$-priceable} if there exists a budget \mbox{$B>0$}
and a payment system $(p_i)_{i\in N}$ satisfying the following constraints:
\begin{itemize}
    \item[\textbf{C1}] $p_i(c) = 0$ if $c \notin A_i$ for all $c \in C$ and $i \in N$
    \item[\textbf{C2}] $\sum_{c \in C} p_i(c) \le \frac{B}{n}$ for all $i \in N$
    \item[\textbf{C3}] $\sum_{i \in N} p_i(c) = 1$ for all $c \in W$
    \item[\textbf{C4}] $\sum_{i \in N} p_i(c) = 0$ for all $c \notin W$ 
    \item[\textbf{C5}] $\sum_{i \in N_c} \left(\frac{B}{n} - \sum_{c \in C} p_i(c)\right) \le 1$ for all $c \notin W$.    
\end{itemize}

We prove a representation guarantee for $B$-priceable committees that are exhaustive. Examples of rules producing such committees include Phragmén's sequential rule \citep{Phra95a}, leximax-Phragmén \citep{BFJL16a}, the Maximin Support Method \citep{SFFB18a}, and Phragmms \citep{CeSt21a}.

\begin{theorem}
    Let $W$ be an exhaustive $B$-priceable committee with budget $B \ge k$. Then, $W$ has a representation guarantee of $\frac{1}{2}$.
    \label{thm:app:pricerep}
\end{theorem}
\begin{proof}
    Assume that $W$ has a representation guarantee of strictly less than $\frac{1}{2}$.  Then, there exists a set of candidates $W'$ with $\cov(W') > 2 \cov(W)$. Let $N'$ be the set of voters covered by $W$ and $N''$ the set of voters covered by $W'$ but not by $W$. We get that $\lvert N'' \rvert > \lvert N' \rvert$.  %
    
     Since $W$ is $B$-priceable and of size $|W|=k$, we know that every voter needs to have at least a budget of $\frac{k}{\lvert N'\rvert}$. However, since $\lvert N'' \rvert > \lvert N' \rvert$  and all voters in $N''$ are covered by $W'$ there must be a candidate in $W'$ which covers more than $\frac{\lvert N' \rvert}{k}$ voters, all previously uncovered. The approvers of this candidate must therefore have a leftover budget of more than $1$, which contradicts the assumption that $W$ is $B$-priceable.  
\end{proof}

\begin{theorem}
    For any $\varepsilon > 0$, there exists an instance in which every exhaustive $B$-priceable committee has a representation ratio of less than $\frac{1}{2} + \varepsilon$.
\end{theorem}
\begin{proof}
    Consider an instance with $\frac{n}{2} + 1$ voters approving $k$ candidates $\{c_1, \dots, c_k\}$ and $k$ groups of size $(\frac{n}{2}-1)/k$ each approving a single candidate 
    not approved by any other group. Now assume that we have an exhaustive $B$-priceable committee $W$ in which less than $k$ candidates from $\{c_1, \dots, c_k\}$ are chosen. Since $W$ is $B$-priceable and contains at least one candidate only approved by $(\frac{n}{2}-1)/k$ voters, we know that $B \ge \frac{k}{\frac{n}{2} - 1}$. Thus, the group of $\frac{n}{2} + 1$ voters have a budget of at least 
    \[
     \left(\frac{n}{2} + 1\right) B \ge \left(\frac{n}{2} + 1\right) \frac{k}{\frac{n}{2} - 1} > k.
    \]
    Therefore, they have enough budget to buy all $k$ candidates from their approval set. Hence, $W$ was not $B$-priceable. 

    In this instance, a $B$-priceable committee can cover at most $\frac{n}{2} + 1$ voters, while the committee optimizing coverage covers $\frac{n}{2} + 1 + (k-1) \frac{\frac{n}{2}-1}{k}$. As for fixed $k$, the representation ratio of this committee approaches $\frac{k}{2k - 1}$ as $n \to \infty$, we get that no representation guarantee of $\frac{1}{2} + \varepsilon$ is possible for $B$-priceable and exhaustive committees.
\end{proof}

\end{document}